\documentclass[a4paper]{amsart}
\usepackage[utf8]{inputenc}
\usepackage{eucal}
\usepackage[english]{babel}
\usepackage{amsmath}
\usepackage{amssymb}
\usepackage{amsopn}                 % defines new operators
\usepackage{amsbsy}
\usepackage{mathtools}              % \MoveEqLeft, \intertext
\usepackage{mathrsfs}
\usepackage{relsize}
\usepackage{enumerate}
\usepackage{enumitem}
\usepackage{comment}
\usepackage{float}
\usepackage[abbreviations]{foreign} % i.e., e.g.
\usepackage{xy}
\xyoption{arrow}
\xyoption{matrix}
\SilentMatrices
\xyoption{tips}
\SelectTips{cm}{10}
\xyoption{2cell}
\UseAllTwocells
\usepackage{tikz}
\usetikzlibrary{cd}
\usetikzlibrary{calc}
\usetikzlibrary{math}
\usetikzlibrary{arrows}
\usetikzlibrary{fit}
\usetikzlibrary{positioning}

\usetikzlibrary{decorations.pathmorphing, arrows.meta}

\tikzset{Rightarrow/.style={double equal sign distance,>={Implies},->},
triple/.style={-,preaction={draw,Rightarrow}},
quadruple/.style={preaction={draw,Rightarrow,shorten >=0pt},shorten >=1pt,-,double,double
distance=0.2pt}}
\usepackage{xcolor}
\definecolor{darkblue}{rgb}{0,0,0.3}
\usepackage[colorlinks=true, citecolor=darkblue, filecolor=darkblue, linkcolor=darkblue, urlcolor=darkblue]{hyperref}
\newtheorem{thm}{Theorem}[section]

\newtheorem{prop}[thm]{Proposition}

\theoremstyle{definition}

\theoremstyle{remark}

{\parskip=0pt\par\nopagebreak\centering}
{\par\noindent\ignorespacesafterend}        % center tikzpictures
\newcommand\nbd\nobreakdash

\newcommand{\Cat}{{\mathcal{C}\mspace{-2.mu}\mathit{at}}}

\newcommand{\nCat}[1]{{#1}\hbox{\protect\nbd-}\kern1pt\Cat}	    % n-categories
\newcommand{\s}{\mathcal{S}\mspace{-2.mu}\text{et}_{\Delta}}

\newcommand{\tr}[2]{\mathchoice
	{#1\raise -1.8pt\vbox{\hbox{$\kern -.8pt/\mathsmaller{#2} $}}}
	{#1\raise -1.8pt\vbox{\hbox{$\kern -.8pt/#2$}}\kern .8pt}
	{#1\raise -1.8pt\vbox{\hbox{$\scriptstyle\kern -.8pt /#2$}}}
	{#1\raise -1.8pt\vbox{\hbox{$\scriptscriptstyle\kern -.8pt /#2$}}}}

\newcommand{\trbis}[2]{\mathchoice
	{#1\raise -1.8pt\vbox{\hbox{$\kern -.8pt\mathsmaller{/#2} $}}}
	{#1\raise -1.8pt\vbox{\hbox{$\kern -.8pt\mathsmaller{/#2}$}}\kern .8pt}
	{#1\raise -1.8pt\vbox{\hbox{$\scriptstyle\kern -.8pt /#2$}}}
	{#1\raise -1.8pt\vbox{\hbox{$\scriptscriptstyle\kern -.8pt /#2$}}}}
\newcommand{\overslice}[2]{\mathchoice
	{#1\raise -1.8pt\vbox{\hbox{$\kern -.8pt\mathsmaller{#2/} $}}}
	{#1\raise -1.8pt\vbox{\hbox{$\kern -.8pt\mathsmaller{#2/}$}}\kern .8pt}
	{#1\raise -1.8pt\vbox{\hbox{$\scriptstyle\kern -.8pt #2/$}}}
	{#1\raise -1.8pt\vbox{\hbox{$\scriptscriptstyle\kern -.8pt #2/$}}}}

	\makeatletter

\def\labelstylecode#1{%
	\pgfkeys@split@path%
	\edef\label@key{/triangle/label/\pgfkeyscurrentname}%
	\edef\style@key{\pgfkeyscurrentkey/.@val}%
	\def\temp@a{#1}%
	\def\temp@b{\pgfkeysnovalue}%
	\ifx\temp@a\temp@b
	\pgfkeysgetvalue{\label@key}\temp@a
	\ifx\temp@a\temp@b\else
	\pgfkeysalso{commutative diagrams/.cd, \style@key}%
	\fi
	\else
	\pgfkeys{\style@key/.code = \pgfkeysalso{#1}}%
	\fi}
\def\arrowstylecode#1{%
	\edef\style@key{\pgfkeyscurrentkey/.@val}%
	\def\temp@a{#1}%
	\def\temp@b{\pgfkeysnovalue}%
	\ifx\temp@a\temp@b
	\pgfkeysalso{commutative diagrams/.cd, \style@key}%
	\else
	\pgfkeys{\style@key/.code = \pgfkeysalso{#1}}%
	\fi}

\pgfkeys{
	/triangle/label/.cd,
	0/.initial = {$\bullet$}, 1/.initial = {$\bullet$}, 2/.initial = {$\bullet$},
	01/.initial, 12/.initial, 02/.initial,
	012/.initial,
	/triangle/labelstyle/.cd,
	01/.@val/.initial=\pgfkeysalso{swap}, 12/.@val/.initial=\pgfkeysalso{swap},
	02/.@val/.initial,
	012/.@val/.initial=\pgfkeysalso{near start},
	01/.code=\labelstylecode{#1}, 02/.code=\labelstylecode{#1},
	12/.code=\labelstylecode{#1}, 012/.code=\labelstylecode{#1},
	/triangle/arrowstyle/.cd,
	01/.@val/.initial, 12/.@val/.initial,
	02/.@val/.initial,
	012/.@val/.initial=\pgfkeysalso{Rightarrow},
	01/.code=\arrowstylecode{#1}, 02/.code=\arrowstylecode{#1},
	12/.code=\arrowstylecode{#1}, 012/.code=\arrowstylecode{#1},
}

\def\tr@abc{%
	\draw [/triangle/arrowstyle/012] (90:0.20) --
	node [/triangle/labelstyle/012] {
		\pgfkeysvalueof{/triangle/label/012}} (270:0.10);
}

\def\tr@#1#2{
	\begin{scope}[shift=#2, commutative diagrams/every diagram]
		
		\node (n{#1}0) at (150:1) {
			\pgfkeysvalueof{/triangle/label/0}};
		\node (n{#1}1) at (270:0.6) {
			\pgfkeysvalueof{/triangle/label/1}};
		\node (n{#1}2) at (30:1) {
			\pgfkeysvalueof{/triangle/label/2}};			
		
		\node (s#1) at (0,0) [circle, inner sep = 0pt,
		fit = (n{#1}0.center)(n{#1}1.center)(n{#1}2.center)] {};
		
		\begin{scope}[commutative diagrams/.cd, every arrow, every label]
			\ifcase #1
			\def\list{0/1, 1/2, 0/2}\or
			\def\list{0/1, 1/2, 0/2}\else
			\def\list{}\fi
			
			\foreach \s / \e in \list {
				\draw [/triangle/arrowstyle/\s\e] (n{#1}\s) --
				node [/triangle/labelstyle/\s\e] {
					\pgfkeysvalueof{/triangle/label/\s\e}} (n{#1}\e);
			}
			
			\ifcase #1
			\tr@abc\or
			\tr@abc
			\else\fi
			
		\end{scope}
	\end{scope}
}

\def\triangle#1{
	\pgfkeys{#1}
	\tr@{0}{(0:0)}
}

\makeatother

\makeatletter

\def\labelstylecode#1{%
	\pgfkeys@split@path%
	\edef\label@key{/square/label/\pgfkeyscurrentname}%
	\edef\style@key{\pgfkeyscurrentkey/.@val}%
	\def\temp@a{#1}%
	\def\temp@b{\pgfkeysnovalue}%
	\ifx\temp@a\temp@b
	\pgfkeysgetvalue{\label@key}\temp@a
	\ifx\temp@a\temp@b\else
	\pgfkeysalso{commutative diagrams/.cd, \style@key}%
	\fi
	\else
	\pgfkeys{\style@key/.code = \pgfkeysalso{#1}}%
	\fi}
\def\arrowstylecode#1{%
	\edef\style@key{\pgfkeyscurrentkey/.@val}%
	\def\temp@a{#1}%
	\def\temp@b{\pgfkeysnovalue}%
	\ifx\temp@a\temp@b
	\pgfkeysalso{commutative diagrams/.cd, \style@key}%
	\else
	\pgfkeys{\style@key/.code = \pgfkeysalso{#1}}%
	\fi}

\pgfkeys{
	/square/label/.cd,
	0/.initial = {$\bullet$}, 1/.initial = {$\bullet$}, 2/.initial = {$\bullet$},
	3/.initial = {$\bullet$},
	01/.initial, 12/.initial, 23/.initial,
	02/.initial, 03/.initial, 13/.initial,
	012/.initial, 013/.initial, 023/.initial, 123/.initial,
	0123/.initial,
	/square/labelstyle/.cd,
	01/.@val/.initial=\pgfkeysalso{swap},
	12/.@val/.initial=\pgfkeysalso{swap},
	23/.@val/.initial=\pgfkeysalso{swap},
	03/.@val/.initial,
	02/.@val/.initial=\pgfkeysalso{description},
	13/.@val/.initial=\pgfkeysalso{description},
	012/.@val/.initial=\pgfkeysalso{above left = -1pt and -1pt},
	013/.@val/.initial=\pgfkeysalso{swap, right},
	023/.@val/.initial=\pgfkeysalso{left},
	123/.@val/.initial=\pgfkeysalso{swap, above right = -1pt and -1pt},
	0123/.@val/.initial,
	01/.code=\labelstylecode{#1}, 02/.code=\labelstylecode{#1},
	03/.code=\labelstylecode{#1}, 12/.code=\labelstylecode{#1},
	13/.code=\labelstylecode{#1}, 23/.code=\labelstylecode{#1},
	012/.code=\labelstylecode{#1}, 013/.code=\labelstylecode{#1},
	023/.code=\labelstylecode{#1}, 123/.code=\labelstylecode{#1},
	0123/.code=\labelstylecode{#1},
	/square/arrowstyle/.cd,
	01/.@val/.initial, 12/.@val/.initial, 23/.@val/.initial,
	02/.@val/.initial, 03/.@val/.initial, 13/.@val/.initial,
	012/.@val/.initial=\pgfkeysalso{Rightarrow},
	013/.@val/.initial=\pgfkeysalso{Rightarrow},
	023/.@val/.initial=\pgfkeysalso{Rightarrow},
	123/.@val/.initial=\pgfkeysalso{Rightarrow},
	0123/.@val/.initial=\pgfkeysalso{Rightarrow},
	01/.code=\arrowstylecode{#1}, 02/.code=\arrowstylecode{#1},
	03/.code=\arrowstylecode{#1}, 12/.code=\arrowstylecode{#1},
	13/.code=\arrowstylecode{#1}, 23/.code=\arrowstylecode{#1},
	012/.code=\arrowstylecode{#1}, 013/.code=\arrowstylecode{#1},
	023/.code=\arrowstylecode{#1}, 123/.code=\arrowstylecode{#1},
	0123/.code=\arrowstylecode{#1}
}

\def\sq@abc{%
	\draw [/square/arrowstyle/012] (235:0.25) --
	node [/square/labelstyle/012] {
		\pgfkeysvalueof{/square/label/012}} (235:0.6);
}
\def\sq@bcd{%
	\draw [/square/arrowstyle/123] (-54:0.25) --
	node [/square/labelstyle/123] {
		\pgfkeysvalueof{/square/label/123}} (-54:0.6);
}
\def\sq@acd{%
	\draw [/square/arrowstyle/023] (55:0.55) --
	node [/square/labelstyle/023] {
		\pgfkeysvalueof{/square/label/023}} (15:0.45);
}
\def\sq@abd{%
	\draw [/square/arrowstyle/013] (125:0.55) --
	node [/square/labelstyle/013] {
		\pgfkeysvalueof{/square/label/013}} (165:0.45);
}

\def\sq@#1#2{
	\begin{scope}[shift=#2, commutative diagrams/every diagram]
		
		\foreach \i in {0,1,2,3} {
			\tikzmath{\a = 135 + (90 * \i);}
			\node (n{#1}\i) at (\a:1) {
				\pgfkeysvalueof{/square/label/\i}};
		}
		
		\node (s#1) at (0,0) [circle, inner sep = 0pt,
		fit = (n{#1}0.center)(n{#1}1.center)(n{#1}2.center)
		(n{#1}3.center)] {};
		
		\begin{scope}[commutative diagrams/.cd, every arrow, every label]
			\ifcase #1
			\def\list{0/1, 1/2, 2/3, 0/2, 0/3}\or
			\def\list{0/1, 1/2, 2/3, 1/3, 0/3}\else
			\def\list{}\fi
			
			\foreach \s / \e in \list {
				\draw [/square/arrowstyle/\s\e] (n{#1}\s) --
				node [/square/labelstyle/\s\e] {
					\pgfkeysvalueof{/square/label/\s\e}} (n{#1}\e);
			}
			
			\ifcase #1
			\sq@abc\sq@acd\or
			\sq@abd\sq@bcd
			\else\fi
			
		\end{scope}
	\end{scope}
}

\def\square#1{
	\pgfkeys{#1}
	\sq@{0}{(180:2)}\sq@{1}{(0:2)}
	
	\begin{scope}[commutative diagrams/.cd, every arrow, every label]
		\draw[->] [shorten >=10pt, shorten <=10pt, /square/arrowstyle/0123] (s0) --
		node [/square/labelstyle/0123] {%
			\pgfkeysvalueof{/square/label/0123}} (s1);
		
	\end{scope}
}

\makeatother

\makeatletter

\def\labelstylecode@pent#1{%
	\pgfkeys@split@path%
	\edef\label@key{/pentagon/label/\pgfkeyscurrentname}%
	\edef\style@key{\pgfkeyscurrentkey/.@val}%
	\def\temp@a{#1}%
	\def\temp@b{\pgfkeysnovalue}%
	\ifx\temp@a\temp@b
	\pgfkeysgetvalue{\label@key}\temp@a
	\ifx\temp@a\temp@b\else
	\pgfkeysalso{commutative diagrams/.cd, \style@key}%
	\fi
	\else
	\pgfkeys{\style@key/.code = \pgfkeysalso{#1}}%
	\fi}
\def\arrowstylecode@pent#1{%
	\edef\style@key{\pgfkeyscurrentkey/.@val}%
	\def\temp@a{#1}%
	\def\temp@b{\pgfkeysnovalue}%
	\ifx\temp@a\temp@b
	\pgfkeysalso{commutative diagrams/.cd, \style@key}%
	\else
	\pgfkeys{\style@key/.code = \pgfkeysalso{#1}}%
	\fi}

\pgfkeys{
	/pentagon/label/.cd,
	0/.initial = {$\bullet$}, 1/.initial = {$\bullet$}, 2/.initial = {$\bullet$},
	3/.initial = {$\bullet$}, 4/.initial = {$\bullet$},
	01/.initial, 12/.initial, 23/.initial, 34/.initial, 04/.initial,
	02/.initial, 03/.initial, 13/.initial, 14/.initial, 24/.initial,
	012/.initial, 013/.initial, 014/.initial, 023/.initial, 024/.initial,
	034/.initial, 123/.initial, 124/.initial, 134/.initial, 234/.initial,
	0123/.initial, 0124/.initial, 0134/.initial, 0234/.initial, 1234/.initial,
	01234/.initial,
	/pentagon/labelstyle/.cd,
	01/.@val/.initial, 12/.@val/.initial, 23/.@val/.initial, 34/.@val/.initial,
	04/.@val/.initial=\pgfkeysalso{swap},
	02/.@val/.initial=\pgfkeysalso{description},
	03/.@val/.initial=\pgfkeysalso{description},
	13/.@val/.initial=\pgfkeysalso{description},
	14/.@val/.initial=\pgfkeysalso{description},
	24/.@val/.initial=\pgfkeysalso{description},
	012/.@val/.initial=\pgfkeysalso{swap, above},
	013/.@val/.initial=\pgfkeysalso{below = 1pt},
	014/.@val/.initial=\pgfkeysalso{below},
	023/.@val/.initial=\pgfkeysalso{swap, above = 1pt},
	024/.@val/.initial=\pgfkeysalso{below left = -1pt and -1pt},
	034/.@val/.initial=\pgfkeysalso{swap, right},
	123/.@val/.initial=\pgfkeysalso{below left = -1pt and -1pt},
	124/.@val/.initial=\pgfkeysalso{swap, right},
	134/.@val/.initial=\pgfkeysalso{left},
	234/.@val/.initial=\pgfkeysalso{swap, below right = -1pt and -1pt},
	0123/.@val/.initial,
	0124/.@val/.initial=\pgfkeysalso{swap},
	0134/.@val/.initial,
	0234/.@val/.initial=\pgfkeysalso{swap},
	1234/.@val/.initial,
	01234/.@val/.initial,
	01/.code=\labelstylecode@pent{#1}, 02/.code=\labelstylecode@pent{#1},
	03/.code=\labelstylecode@pent{#1}, 04/.code=\labelstylecode@pent{#1},
	12/.code=\labelstylecode@pent{#1}, 13/.code=\labelstylecode@pent{#1},
	14/.code=\labelstylecode@pent{#1}, 23/.code=\labelstylecode@pent{#1},
	24/.code=\labelstylecode@pent{#1}, 34/.code=\labelstylecode@pent{#1},
	012/.code=\labelstylecode@pent{#1}, 013/.code=\labelstylecode@pent{#1},
	014/.code=\labelstylecode@pent{#1}, 023/.code=\labelstylecode@pent{#1},
	024/.code=\labelstylecode@pent{#1}, 034/.code=\labelstylecode@pent{#1},
	123/.code=\labelstylecode@pent{#1}, 124/.code=\labelstylecode@pent{#1},
	134/.code=\labelstylecode@pent{#1}, 234/.code=\labelstylecode@pent{#1},
	0123/.code=\labelstylecode@pent{#1}, 0124/.code=\labelstylecode@pent{#1},
	0134/.code=\labelstylecode@pent{#1}, 0234/.code=\labelstylecode@pent{#1},
	1234/.code=\labelstylecode@pent{#1}, 01234/.code=\labelstylecode@pent{#1},
	/pentagon/arrowstyle/.cd,
	01/.@val/.initial, 12/.@val/.initial, 23/.@val/.initial, 34/.@val/.initial,
	04/.@val/.initial, 02/.@val/.initial, 03/.@val/.initial, 13/.@val/.initial,
	14/.@val/.initial, 24/.@val/.initial,
	012/.@val/.initial=\pgfkeysalso{Rightarrow},
	013/.@val/.initial=\pgfkeysalso{Rightarrow},
	014/.@val/.initial=\pgfkeysalso{Rightarrow},
	023/.@val/.initial=\pgfkeysalso{Rightarrow},
	024/.@val/.initial=\pgfkeysalso{Rightarrow},
	034/.@val/.initial=\pgfkeysalso{Rightarrow},
	123/.@val/.initial=\pgfkeysalso{Rightarrow},
	124/.@val/.initial=\pgfkeysalso{Rightarrow},
	134/.@val/.initial=\pgfkeysalso{Rightarrow},
	234/.@val/.initial=\pgfkeysalso{Rightarrow},
	0123/.@val/.initial=\pgfkeysalso{triple},
	0124/.@val/.initial=\pgfkeysalso{triple},
	0134/.@val/.initial=\pgfkeysalso{triple},
	0234/.@val/.initial=\pgfkeysalso{triple},
	1234/.@val/.initial=\pgfkeysalso{triple},
	01234/.@val/.initial=\pgfkeysalso{quadruple},
	01/.code=\arrowstylecode@pent{#1}, 02/.code=\arrowstylecode@pent{#1},
	03/.code=\arrowstylecode@pent{#1}, 04/.code=\arrowstylecode@pent{#1},
	12/.code=\arrowstylecode@pent{#1}, 13/.code=\arrowstylecode@pent{#1},
	14/.code=\arrowstylecode@pent{#1}, 23/.code=\arrowstylecode@pent{#1},
	24/.code=\arrowstylecode@pent{#1}, 34/.code=\arrowstylecode@pent{#1},
	012/.code=\arrowstylecode@pent{#1}, 013/.code=\arrowstylecode@pent{#1},
	014/.code=\arrowstylecode@pent{#1}, 023/.code=\arrowstylecode@pent{#1},
	024/.code=\arrowstylecode@pent{#1}, 034/.code=\arrowstylecode@pent{#1},
	123/.code=\arrowstylecode@pent{#1}, 124/.code=\arrowstylecode@pent{#1},
	134/.code=\arrowstylecode@pent{#1}, 234/.code=\arrowstylecode@pent{#1},
	0123/.code=\arrowstylecode@pent{#1}, 0124/.code=\arrowstylecode@pent{#1},
	0134/.code=\arrowstylecode@pent{#1}, 0234/.code=\arrowstylecode@pent{#1},
	1234/.code=\arrowstylecode@pent{#1}, 01234/.code=\arrowstylecode@pent{#1}
}

\def\pent@abc{%
	\draw [/pentagon/arrowstyle/012] (198:0.45) --
	node [/pentagon/labelstyle/012] {
		\pgfkeysvalueof{/pentagon/label/012}} (198:0.8);
}
\def\pent@bcd{%
	\draw [/pentagon/arrowstyle/123] (126:0.45) --
	node [/pentagon/labelstyle/123] {
		\pgfkeysvalueof{/pentagon/label/123}} (126:0.8);
}
\def\pent@cde{%
	\draw [/pentagon/arrowstyle/234] (54:0.45) --
	node [/pentagon/labelstyle/234] {
		\pgfkeysvalueof{/pentagon/label/234}} (54:0.8);
}
\def\pent@ade{%
	\draw [/pentagon/arrowstyle/034] (-40:0.6) --
	node [/pentagon/labelstyle/034] {
		\pgfkeysvalueof{/pentagon/label/034}} (-5:0.5);
}
\def\pent@abe{%014
	\draw [/pentagon/arrowstyle/014] (-70:0.55) --
	node [/pentagon/labelstyle/014] {
		\pgfkeysvalueof{/pentagon/label/014}} (-110:0.55);
}
\def\pent@acd{%
	\draw [/pentagon/arrowstyle/023] (55:0.3) --
	node [/pentagon/labelstyle/023] {
		\pgfkeysvalueof{/pentagon/label/023}} (125:0.3);
}
\def\pent@bde{%
	\draw [/pentagon/arrowstyle/134] (-5:0.4) --
	node [/pentagon/labelstyle/134] {
		\pgfkeysvalueof{/pentagon/label/134}} (35:0.5);
}
\def\pent@ace{%
	\draw [/pentagon/arrowstyle/024] (-45:0.45) --
	node [/pentagon/labelstyle/024] {
		\pgfkeysvalueof{/pentagon/label/024}} (-45:0.1);
}
\def\pent@abd{%
	\draw [/pentagon/arrowstyle/013] (-90:0.22) --
	node [/pentagon/labelstyle/013] {
		\pgfkeysvalueof{/pentagon/label/013}} (-150:0.46);
}
\def\pent@bce{%
	\draw [/pentagon/arrowstyle/124] (188:0.4) --
	node [/pentagon/labelstyle/124] {
		\pgfkeysvalueof{/pentagon/label/124}} (150:0.55);
}

\def\pent@#1#2{
	\begin{scope}[shift=#2, commutative diagrams/every diagram]
		
		\foreach \i in {0,1,2,3,4} {
			\tikzmath{\a = 270 - (72 * \i);}
			\node (n{#1}\i) at (\a:1) {
				\pgfkeysvalueof{/pentagon/label/\i}};
		}
		
		\node (p#1) at (0,0) [circle, inner sep = 0pt,
		fit = (n{#1}0.center)(n{#1}1.center)(n{#1}2.center)
		(n{#1}3.center)(n{#1}4.center)] {};
		
		\begin{scope}[commutative diagrams/.cd, every arrow, every label]
			\ifcase #1
			\def\list{0/1, 1/2, 2/3, 3/4, 0/4, 0/2, 0/3}\or
			\def\list{0/1, 1/2, 2/3, 3/4, 0/4, 1/3, 1/4}\or
			\def\list{0/1, 1/2, 2/3, 3/4, 0/4, 0/2, 2/4}\or
			\def\list{0/1, 1/2, 2/3, 3/4, 0/4, 0/3, 1/3}\or
			\def\list{0/1, 1/2, 2/3, 3/4, 0/4, 1/4, 2/4}\else
			\def\list{}\fi
			
			\foreach \s / \e in \list {
				\draw [/pentagon/arrowstyle/\s\e] (n{#1}\s) --
				node [/pentagon/labelstyle/\s\e] {
					\pgfkeysvalueof{/pentagon/label/\s\e}} (n{#1}\e);
			}
			
			\ifcase #1
			\pent@abc\pent@acd\pent@ade\or
			\pent@bcd\pent@bde\pent@abe\or
			\pent@cde\pent@ace\pent@abc\or
			\pent@ade\pent@abd\pent@bcd\or
			\pent@abe\pent@bce\pent@cde
			\else\fi
			
		\end{scope}
	\end{scope}
}

\def\pentagon#1{
	\pgfkeys{#1}
	\pent@{2}{(270:3)}\pent@{0}{(198:3)}\pent@{3}{(126:3)}
	\pent@{1}{(54:3)}\pent@{4}{(342:3)}
	
	\begin{scope}[commutative diagrams/.cd, every arrow, every label]
		\draw [/pentagon/arrowstyle/0123] (p0) --
		node [/pentagon/labelstyle/0123] {
			\pgfkeysvalueof{/pentagon/label/0123}} (p3);
		
		\draw [/pentagon/arrowstyle/0134] (p3) --
		node [/pentagon/labelstyle/0134] {
			\pgfkeysvalueof{/pentagon/label/0134}} (p1);
		
		\draw [/pentagon/arrowstyle/1234] (p1) --
		node [/pentagon/labelstyle/1234] {
			\pgfkeysvalueof{/pentagon/label/1234}} (p4);
		
		\draw [/pentagon/arrowstyle/0234] (p0) --
		node [/pentagon/labelstyle/0234] {
			\pgfkeysvalueof{/pentagon/label/0234}} (p2);
		
		\draw [/pentagon/arrowstyle/0124] (p2) --
		node [/pentagon/labelstyle/0124] {
			\pgfkeysvalueof{/pentagon/label/0124}} (p4);
		
		\draw [/pentagon/arrowstyle/01234] (270:0.75) --
		node [/pentagon/labelstyle/01234] {
			\pgfkeysvalueof{/pentagon/label/01234}} (90:0.75);
	\end{scope}
}

\makeatother
\usepackage{url}
\usepackage{}
\usepackage{ragged2e}
\usepackage{hyperref}
\usepackage{pdfpages}
\newcommand{\R}{\mathbb{R}}

\newtheorem{problem}[thm]{Problem}
\newtheorem{rmk}[thm]{Remark}

\setlist[itemize]{leftmargin=*}
\setlist[enumerate]{leftmargin=*}

\makeatletter

\renewcommand{\tocsection}[3]{%
\indentlabel{\@ifnotempty{#2}{\parbox[b]{3ex}{\bfseries\ignorespaces#1 #2}}}\bfseries#3} 

\renewcommand{\tocsubsection}[3]{%
\indentlabel{\@ifnotempty{#2}{\hspace{1.6em}\parbox[b]{5ex}{\ignorespaces#1 #2}}}#3}

\makeatother

\title{Fire sales, the LOLR and bank runs with continuous asset liquidity}

\author{Ulrich Bindseil}
\address{European Central Bank, DG-Market Infrastructure and Payments}
\email{ulrich.bindseil@ecb.europa.eu}
\author{Edoardo Lanari}
\address{Institute of Mathematics, Czech Academy of Sciences \\ \v{Z}itn\'a 25 \\115 67   Praha 1\\ Czech Republic}
\email{edoardo.lanari.el@gmail.com}
\urladdr{https://edolana.github.io/}
\date{\today}
\begin{document}

\begin{abstract}
	Bank’s asset fire sales and recourse to central bank credit are modelled with continuous asset liquidity, allowing to derive the liability structure of a bank. Both asset sales liquidity and the central bank collateral framework are modeled as power functions within the unit interval.  Funding stability is captured as a strategic bank run game in pure strategies between depositors. Fire sale liquidity and the central bank collateral framework  determine jointly the ability of the banking system to deliver maturity transformation without endangering financial stability. The model also explains why banks tend to use the least liquid eligible collateral with the central bank and why a sudden non-anticipated reduction of asset liquidity, or a tightening of the collateral framework, can trigger a bank run. The model also shows that the collateral framework can be understood, beyond its aim to protect the central bank, as financial stability and non-conventional monetary policy instrument. 
\end{abstract}
\maketitle
\makeatletter
\def\blfootnote{\gdef\@thefnmark{}\@footnotetext}
\makeatother%
\blfootnote{Opinions expressed are those of the authors. We wish to thank Steffano Corradin, Philipp Harms, Florian Heider, Slobodan Jelic, Simone Manganelli, Andres Manzanares, Fernando Monar, Ken Nyholm, Isabel Schnabel, Leo von Thadden, and an anonymous referee. The second named author gratefully acknowledges the support of Praemium Academiae of M.~Markl and RVO:67985840.}
\section{Introduction}
The model proposed in this paper sheds new light on how asset liquidity and the central bank collateral framework affect the liability structure of banks, financial stability and monetary policy. The financial crisis of 2007/2008 is said to also have been triggered by the insufficient asset liquidity buffers of banks relative to their short-term liabilities. These insufficient buffers would have led to an (at least temporarily) excessive reliance on central bank funding (\cite{Basel}). The economic tradeoffs between the efficiency of the banking system in delivering maturity transformation and financial stability is also crucial when assessing the net benefits of regulation for society. In the words of the Turner review (\cite{FSA}): \\
\newline
\textit{``[T]here is a tradeoff to be struck. Increased maturity transformation delivers benefits to the non-bank sectors of the economy and produces term structures of interest rates more favourable to long-term investment. But the greater the aggregate degree of maturity transformation, the more the systemic risks and the greater the extent to which risks can only be offset by the potential for central bank liquidity assistance.''}\\
\newline
While the central bank collateral framework got relatively limited attention in academic writing, it is one of the most complex and economically significant elements of monetary policy implementation(e.g. \cite{Nyborg}, \cite{Bindseil2017}, \cite{Bindseil2017-2}). Unencumbered central bank eligible collateral is potential liquidity, as it can, in principle, be swapped into central bank money. It is therefore not exaggerated to argue that the collateral framework must be an important ingredient of any theory of liquidity crises (as noted in \cite{Bagehot}), and of any monetary theory. The Markets Committee of the BIS summarizes in \cite{Markets} various measures taken during the Lehman crisis by central banks (p 8-9):\\
\newline
\textit{``During the height of the financial crisis in 2008–09, a number of central banks introduced, to varying degrees, crisis management measures such as a temporary acceptance of additional types of collateral, a temporary lowering of the minimum rating requirements of existing eligible collateral or a temporary relaxation of haircut standards.''}\\
\newline
The COVID-2019 crisis again led central banks to relax their collateral framework, despite the fact that credit quality of issuers and counterparties did not improve. The ECB announced an “unprecedented” package on 7 April 2020 (from ECB Press release, \cite{ECB20}):\\
\newline
\textit{``The Governing Council of the European Central Bank (ECB) today adopted a package of temporary collateral easing measures to facilitate the availability of eligible collateral for Eurosystem counterparties to participate in liquidity providing operations… the Eurosystem is increasing its risk tolerance to support the provision of credit via its refinancing operations, particularly by lowering collateral valuation haircuts for all assets consistently.  … the Governing Council decided to temporarily increase its risk tolerance level in credit operations through a general reduction of collateral valuation haircuts by a fixed factor of 20\%.''}\\
\newline
Three strands of academic literature are relevant to the present paper. First, Rochet and Vives (\cite{Rochet}) is close to the present paper in the sense that it also models the role of fire sales and the central bank LOLR for banks’ funding stability, liquidity and solvency. The model of Rochet and Vives, however, takes strong simplifying assumptions regarding asset liquidity (only two types of assets are distinguished: cash and non-liquid assets). Beyond Rochet and Vives, there is an extensive more general multiple funding equilibrium literature such as represented by e.g. Morris and Shin (\cite{Morris}) under the headline of “global games”. This literature uses more general and sophisticated equilibrium concepts than the present paper, which limits itself to pure and dominant strategies of investors/depositors, and to the existence or not of a Strict Nash Equilibrium in the sense of Fudenberg and Tirole (\cite{Fudenberg}).\footnote{Other relevant papers that model funding liquidity, leverage and asset liquidity are Brunnermeier and Pedersen (\cite{Brunnermeier}), Acharya, Gale and Yorulmazer (\cite{Acharya-Gale}), and Acharya and Viswanathan (\cite{Acharya}).} Again, none of these papers however models explicitly fire sales and the central bank LOLR together, nor do they allow to derive the bank’s liability structure from this. Finally, there is a relatively recent bank run literature integrating bank runs and the LOLR into macro-economic models, such as Kiyotaki and Gertler (\cite{Gertler}) and De Fiore, Hoerova and Uhlig (\cite{DeFiore}). These models are more ambitious in terms of integrating bank runs in a general equilibrium macro context, but the modelling of the limits of banks to issue short term deposits is rather ad hoc (a moral hazard problem of the bank manager is postulated to limit deposit issuance, namely bank managers would be able to divert asset from the bank to enrich themselves personally). In contrast, in the model proposed here, the bank’s ability to issue short term deposits follows endogenously from the threat of a bank run. 
\\

Second, Ashcraft et al (\cite{Ashcraft-Garleanu}) relates to the present model in the sense that central bank haircut policies are identified and modeled as a monetary policy instrument (see also Chapman et al, \cite{Chapman}). Ashcraft et al assume that banks refinance assets at the central bank and that the haircut determines the leverage ratio and thus the funding costs of assets, being a weighted average of the risk free rate and the shadow cost of equity (see also Brunnermeier and Pedersen, \cite{Brunnermeier}).\\

Third, the present paper explains the spread between the risk free rate (which is close to the rate of central bank credit operations and the rate of remuneration of overnight deposits of households with the banking system) and the actual funding costs of the real economy (or the effective monetary conditions for the economy). In this sense, the paper contributes to the program such as defined for instance by Woodford (\cite{Ashcraft-Garleanu}) to enrich the analysis of monetary policy in particular by capturing more explicitly the spread between the central bank credit operations rate and the actual monetary conditions as they are felt by the real economy when seeking bank or market funding. Indeed, the present paper shows that a drop in bank asset liquidity and/or an increase in central bank haircuts both tighten effective monetary conditions in the sense that they reduce the ability of the banking system to undertake maturity transformation, and hence, everything else equal, will increase the share of “expensive” bank funding sources such as long term bonds and equity, implying that also the lending rates that a competitive banking system is able to offer have to increase.\\

Empirical evidence of the different degrees of banks’ asset liquidity, focusing on the case of the euro area, is provided by Bindseil, Dragu, Düring and von Landesberger (\cite{Bindseil2017-2}). The paper provides a cross-section analysis of liquidity properties and central bank collateral haircuts of the euro area fixed income universe. 
That asset liquidity is continuous, and that it fluctuates over time, has been described empirically in the finance literature, such as recently in D\"otz and Weth (\cite{Dotz}), who also argue that liquidation will be carried out in a liquidity pecking order style and that marginal liquidation costs increase in redemptions (p. 9-11), i.e. as assumed in the present paper. They construct a sample of corporate bond fund asset liquidity data covering the 80 months before June 2016, referring to around 700 thousand security holdings positions (p. 12). The liquidity measure consists in monthly averages of bid-ask spreads. The following figure shows continuous portfolio liquidity, put at any moment in time in a “liquidity pecking order” (i.e. securities ranked from the most to the least liquid). Obviously, the least liquid assets held by a corporate bond fund will still be more liquid than many other bank assets, such as in particular loan portfolios. Still, D\"otz and Weth illustrates the idea of continuous asset liquidity and the changes of asset liquidity over time. 
\begin{figure}[H]
\includegraphics[scale = 0.8]{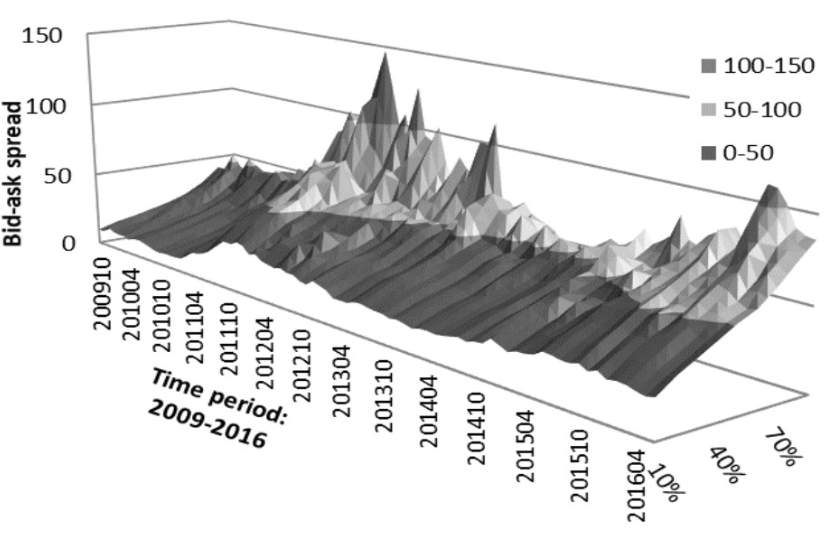}
\caption{Liquidity structure of corporate bond funds, according to D\"otz and Weth \cite{Dotz}}
\end{figure}
The paper also sheds new light on a number of other related debates. For example, first, Acosta-Smith et al (\cite{Aco}), who study the relationship between bank capital and liquidity transformation find that ``banks engage in less liquidity transformation when they have higher capital'', but our paper suggests that there are alternative causalities that could explain this empirical pattern. Second, the literature on fire sales, as summarized by Shleifer and Vishny (\cite{Shlei}), generally regards fire sales as entailing ``systemic risk and significant [negative] externalitie''(p. 43). Our paper models potential fire sales as being an integral part of maturity transformation by banks, whereby indeed, actual fire sales would not materialize in the absence of negative shocks on asset liquidity or asset values. Third, the present model provides further insights into empirical phenomena, such as those observed by Boyson et al (\cite{boy}), which investigate ``which factor, liquidity or solvency, is more important for financial crisis''. Several further empirical hypothesis supported in that paper are captured also in our model, and therefore are supported by additional theoretical explanations. Finally, the paper also sheds new light on why both liquidity regulation, and the lender of last resort are needed, such as discussed by Carlson et al (\cite{carl}). 
\section{A model of funding stability with continuous asset liquidity and haircuts}
Throughout this paper, the following stylized bank balance sheet is assumed, with  total length  set to unit. Assets are heterogeneous in a continuous sense, while there are three types of liabilities which are each separately homogenous equity (\(e\)), long term debt (\(t\)), and short term deposits held by two depositors (each depositor holding therefore \((1-t-e)/2\) short term deposits), with \(e,t \geq 0, \ e+t \in [0,1]\).
\begin{figure}[H]
\caption{Stylised bank balance sheet}
\begin{tabular}{|llc|crr|}
\hline
&\text{Assets} &&\text{Liabilities} &&\\
\hline
\text{Assets} & & 1& \text{Short term debt 1} && (1-t-e)/2\\
&& & \text{Short term debt 2} && (1-t-e)/2\\
&&& \text{Long term debt (``term funding'')} && t\\
&&& \text{equity} && e\\		
\hline
\end{tabular}
\end{figure} 
Short term deposits can be subject to a bank run (Diamond and Dybvig, \cite{Diamond}).  Banks could address this by not relying, or at least not extensively relying on short term funding. However, in general investors prefer to hold short term debt instruments over long term debt instruments, and hence require a higher interest rate on long term debt. In other words, long term debt is associated with higher funding costs for the bank. Banks could also  hold sufficient amounts of liquid assets, both in the sense of being able to liquidate these assets in case of need, and in order to be able to pledge them with the central bank at limited haircuts. However, on average, liquid assets generate lower return than illiquid ones. Consider now in more detail the different balance sheet positions.
\subsection{Bank assets}
Assets are continuous with regard to two liquidity characteristics: (i) Asset liquidity as measured by the “fire sale” discount to be accepted if an asset is to be sold in the short run; (ii) valuation haircut if submitted as central bank collateral.
\paragraph{\textbf{Bank assets as central bank collateral}}
Assume that assets are ranked from those which the central bank considers to require the lowest haircuts to those requiring the highest ones.  The central bank collateral haircut function is set to be from the asset unit interval \([0,1]\) into the possible haircut unit interval \([0,1]\). Assume that it has the following functional form with \(\delta \geq 0\):
\begin{equation}
	h(x)=x^{\delta}
\end{equation}
The power function in the unit interval captures broadly the properties of a typical central bank haircut framework: haircuts for the most liquid assets will be close to zero, while haircuts for the least liquid assets accepted will be high, and an often significant part of assets will not be accepted at all, which is equivalent to a \(100\%\) haircut. If \(\delta\) is close to 0, then the haircuts increase and converge quickly towards 1. If in contrast \(\delta\) is large (say 10) then haircuts stay at close to zero for a while and only start to increase in a convex manner when approaching the least liquid assets. The total haircut (and the average haircut) if all assets are pledged is \(\frac{1}{\delta +1}\), and hence potential central bank credit is \(\frac{\delta}{\delta +1}\). This is obtained from the integration rule \(\int x^{\delta}dx = \frac{1}{\delta +1}x^{\delta + 1}\). For example, in the case of the Eurosystem, out of EUR 30 trillion of aggregated bank assets, the value of central bank eligible collateral after haircuts that could be used at any moment in time is around EUR 5 trillion. The eligibility criteria and haircut matrices are provided by the ECB and one can match this information in principle with an informed guess of banks’ assets holdings. This implies that the effective average haircut applied by the Eurosystem to (the entirety of) bank assets is around 83\% (25 trillion /30 trillion), and central bank refinancing power is around 17\% of eligible assets, which approximately implies, if one assumes a power function as done above,  a parameter value \(\delta=0.2\). ECB (2015,31) provides an overview of the ECB’s haircut scheme.
\footnote{\href{https://www.ecb.europa.eu/pub/pdf/other/financial_risk_management_of_eurosystem_monetary_policy_operations_201507.en.pdf?b8a471f2b2dcf413890c1d02f7288648}{Link: financial risk management of eurosystem monetary policy operations}}
\paragraph{\textbf{Liquidity of bank assets}}
Now consider asset liquidity in the sense of the ability of banks to sell assets in the short term without this inflicting a loss for the bank. Assume again that assets are ranked from the most liquid to the least liquid, and that the fire sale discount function is a function from \([0,1]\) to \([0,1]\) with the following function form, for \(\Theta \geq 0\):
\begin{equation}
d(x)=x^{\Theta}
\end{equation}
The smaller \(\Theta\), the faster fire sale losses increase towards 100\% when moving from the most liquid to the least liquid assets.  If a certain share \(x\) of the assets has to be sold, then the fire sale discounts will have to be booked as a loss and reduce equity. Assuming that the bank starts with the most liquid assets, the loss will be \(\frac{1}{\Theta + 1} x^{\Theta + 1}\). Empirical estimates of default costs in the corporate finance literature vary between 10 and 44\% (see e.g. Glover, \cite{Glover}, and Davydenko et al, \cite{Davy}). This cost can be interpreted as the liquidation cost of assets, captured in the parameter \(\Theta\). Liquidation of all assets will lead to a damage of \(\frac{1}{\Theta + 1}\), so that the remaining asset value will be \(\frac{\Theta}{\Theta + 1}\). If default cost is 10\%, this would mean that \(\Theta = 9\), and if default cost is 44\%, then \(\Theta = 1.27\). For a value of default costs in the middle of the empirical estimates of say 25\%, one obtains \(\Theta = 3\). 
\subsection{Bank liabilities}
Four types of liabilities are distinguished: (i) \textit{Short term liabilities} are equally split to two ex-ante identical depositors; (ii) \textit{Long term debt} does not mature within the period considered, and are ranked pari passu with short term debt; (iii) \textit{Equity} is junior to all other liabilities, and cannot flow out either; (iv) \textit{Central bank borrowing} is zero initially, but can substitute for outflows of short term liabilities in case of need. It is collateralized and therefore the central bank acquires in case of default ownership of the assets pledged as collateral. Apart from this, the central bank claim ranks pari passu, i.e. remaining claims after collateral liquidation are treated in the same way as an initial unsecured deposit. 
\subsection{Timeline}
The model is based on the following timeline:
\begin{itemize}
\item Initially, the bank has the balance sheet composition as shown in figure 1.
\item Short term depositors/investors play a strategic game with two alternative actions: to run or not to run. ``Running'' means to withdraw the deposits. If successful, this means that the value of deposits is afterwards equal to the initial value minus a small cost, \(\epsilon\), capturing the transaction cost of withdrawing the deposits.
\item It is not to be taken for granted that depositors can withdraw all their funds. Outflows need to be funded by the bank in some way: (i) Recourse to central bank credit, assuming that the bank has sufficient eligible collateral. (ii) Quick liquidation of assets (``fire sales'').  (iii) If it is impossible to pay out the depositors that want to withdraw their deposits, Illiquidity induced default will occur. If illiquidity induced default occurs, all (remaining) assets need to be liquidated, and corresponding default related losses occur. 
\item If the bank was not closed due to illiquidity, still its solvency needs verification after it had to undertake fire sales and if capital is negative, the bank is resolved. Again, it is assumed in this case that the full costs of immediately liquidating all assets materialize. 
\end{itemize}
% \begin{figure}[b]
% \caption{Timeline of bank liability structure model}
% \includegraphics[width=1.3\textwidth]{figure2.png}
% \label{timeline figure}
% \end{figure}
\subsection{Strict Nash No-Run (SNNR) equilibrium}
The decision set of depositor \(i=1,2\) from which he will choose his decision \(D_i\)  consists in \(\{K_i,R_i\}\), whereby ``\(K\)'' stands for ``keeping''  deposits and ``R'' stands for ``run''. Let \(U_i=U_i(D_1,D_2)\) be the pay-off function of depositor \(i\). Note that the strategic game is symmetric, i.e. \(U_1(K_1,K_2)=U_2(K_1,K_2), \ U_1(K_1,R_2)=U_2(R_1,K_2), \ U_1(R_1,K_2)=U_2(K_1,R_2), \ U_1(R_1,R_2)=U_2(R_1,R_2)\). This allows to express in the rest of the paper conditions only with reference to one of the two players, say depositor 1. 

A Strict Nash equilibrium is defined as a strategic game in which each player has a unique best response to the other players’ strategies (see Fudenberg and Tirole, \cite{Fudenberg}). A Strict Nash No-Run (SNNR) equilibrium in the run game is therefore one in which the no-run choice dominates the ``run'' choice regardless of what the other depositors decide, i.e. an SNNR equilibrium is defined by:
\begin{equation}
U_1(K_1,K_2)>U_1(R_1,K_2) \quad \text{and} \quad U_1(K_1,R_2)>U_1(R_1,R_2)	 
\end{equation}
\section{Pure reliance on either central bank credit or on asset fire sales}
\subsection{Pure reliance on central bank funding}
Assume first that asset liquidation is not an option, say because markets are totally frozen, i.e. \(\Theta =0\). In this case the analysis can focus on the sufficiency or not of buffers for central bank credit. The following proposition states the necessary condition for funding stability of banks in this case. 
\begin{prop}
\label{prop1}
If \(\Theta =0\), and assuming a small transaction cost \(\epsilon\) of withdrawing deposits, a Strict Nash No-Run (SNNR) equilibrium prevails if and only if \(\frac{\delta}{\delta+1}\geq \frac{(1-t-e)}{2}\), i.e. the liquidity buffer based on recourse to the central bank is not smaller than one half of the short term deposits.	
\end{prop}
\begin{proof}
 To prove this result (and similar subsequent propositions), it is sufficient to calculate through the pay-offs for the alternative decisions of depositors under the possible parameter combinations and establish the frontiers of parameter combinations under which the conditions of an SNNR equilibrium apply. Distinguish now the three possible cases:  \((1-t-e) < \frac{\delta}{\delta+1} \)(which we will denote by \((1)\)), \( \frac{(1-t-e)}{2} < \frac{\delta}{\delta+1} < (1-t-e)\) (which we will denote by \((2)\)) and\( \frac{\delta}{\delta+1}< \frac{(1-t-e)}{2}\) (denoted by \((3)\)). These represent, respectively, the cases in which liquidity buffers provided by central bank collateral are sufficient to compensate for the withdrawal of all deposits, for the withdrawal of (not more than) one depositor or for not even the withdrawal of one depositor. It may be noted that in the case assumed here that \(\Theta =0\), illiquidity of the bank means complete destruction of asset value as the liquidation value of assets is assumed to be zero. The proposition follows from verifying the condition \(U_1(K_1,K_2)>U_1(R_1,K_2) \quad \text{and} \quad U_1(K_1,R_2)>U_1(R_1,R_2) \) for the three different cases distinguished above. One obtains the following scenarios:
 \[\begin{matrix}
 \text{Case} & U_1(K_1,K_2) && U_1(R_1,K_2) &&U_1(K_1,R_2) &&U_1(R_1,R_2)\\ && && && &&\\

 (1)& \frac{(1-t-e)}{2}&&\frac{(1-t-e)}{2}-\epsilon &&\frac{(1-t-e)}{2}&&\frac{(1-t-e)}{2} –\epsilon \\
 &&&&&&&& \\

(2)& \frac{(1-t-e)}{2}&& \frac{(1-t-e)}{2}-\epsilon && \frac{(1-t-e)}{2}&& \frac{1}{2}\frac{\delta}{\delta+1}\\
&&&&&&&& \\
 (3)& \frac{(1-t-e)}{2}&& \frac{\delta}{\delta+1}&& 0 &&\frac{1}{2}\frac{\delta}{\delta+1}\\
 \end{matrix}\]
 It is easily verified that the conditions for an SNNR equilibrium are given if and only if \(\frac{(1-t-e)}{2}<\frac{\delta}{\delta+1}\).
 \end{proof} 
\subsection{Pure reliance on asset fire sales}
Now consider the case in which the central bank does not offer any credit to banks, i.e. accepts no collateral at all. In this case \(\delta=0\), such that addressing deposit outflows will have to rely exclusively on asset liquidation. Assume that the bank does whatever it takes in terms of asset liquidation to avoid illiquidity induced default. The total amount of liquidity that the bank can generate through asset fire sales is \(\frac{\Theta}{\Theta+1}\). While with full reliance on central bank lending, the question was whether the related liquidity buffers would be sufficient. In the present case, two default triggering events need to be considered. Indeed, even if the bank has survived a liquidity withdrawal, it may afterwards be assessed as insolvent and thus be liquidated at the request of the bank supervisor. As noted above, for a given liquidity withdrawal \(x\), the fire sale related loss is \(\frac{1}{\Theta+1}x^{\Theta+1}\). The latter default event occurs if this loss exceeds initial equity\footnote{Note that it is assumed that equity is never sufficient to absorb the losses resulting from a bank default, i.e. it is assumed that \(e \leq 1/(\Theta+1)\). Of course one could also calculate through the opposite case, but it is omitted here as it does not seem to match reality.}. 
\begin{prop}
If \(\delta = 0\), a SNNR equilibrium exists if and only if \[\frac{(1-t-e)}{2} \leq \frac{\Theta}{(\Theta+1)} \  \ \ \text{and} \ \ \  e\geq
\frac{1}{(\Theta+1)}\frac{(1-t-e)}{2}^{\Theta+1}.\]	
\end{prop}
The proposition can be verified by again establishing the strategic game pay-offs and showing under which circumstances the SNNR conditions are met. The proof is provided in \cite{Bindseil}. In sum, to ensure financial stability in the case of absence of central bank credit, minimum liquidity and capital buffers are needed in some appropriate combination to ensure the stability of a given amount of short-term funding. The lower the asset liquidity, the lower the amount of short-term funding that can be sustained for a given level of equity.
\section{Cases in which the banks rely on both types of liquidity buffers}
Now consider the cases in which both \(\Theta,\delta >0\). It is assumed that the ordering of assets is the same for both forms of liquidity generation, i.e. if asset \(i\) is subject to lower fire sale discounts than asset \(j\), then also asset \(i\) will have a lower central bank collateral haircut than asset \(j\). Proposition \ref{prop 3} narrows down the actual range of mixed cases, i.e. cases in which both liquidity sources play a role in the planning of the bank. 
\begin{prop}
\label{prop 3}
If either \(\delta \geq \Theta\), or both \(\Theta >\delta\) and \(\delta/(\delta+1) \geq (1-t-e)/2\) are satisfied, then banks will only rely on central bank credit to address possible deposit withdrawals, and hence the conditions established in Proposition \ref{prop1} apply to the existence of an SNNR equilibrium. 
\end{prop}
Taking recourse to the central bank does not cause a loss, while fire sales cause one. If in addition, central bank recourse yields more liquidity (i.e. \(\delta >\Theta\)), then central bank credit strictly dominates asset fire sales as a source of emergency liquidity. If \(\Theta > \delta\) and central bank liquidity buffers allow to address liquidity outflows relating to one depositor, i.e. \(\delta/(\delta +1) \geq (1-t-e)/2\), which, as shown previously, allows to sustain the SNNR equilibrium, then again relying only on central bank credit dominates strategies to rely on both sources. 

The cases in which the bank wants to rely potentially on both funding sources therefore appear to be limited to the ones in which \(\Theta > \delta\) and \((1-t-e)/2>\delta/(\delta+1)\). Again a number of cases have to be distinguished. There will generally be a trade-off between the maximum liquidity generation and the ability to avoid losses, under the optimal use of the two funding sources. For example, the maximum generation of liquidity is achieved through fire sales only, and will be equal to \(\Theta/(1+\Theta)\). However, this also leads to the highest possible fire sales losses and damage to equity, \(1/(1+\Theta)\), and it is realistic to assume that this extent of losses would exceed equity, and anyway if all of the assets of the bank are sold, it has ceased to exist. The lowest generation of liquidity is achieved if all assets are pledged through central bank credit, and in this case liquidity generation is \(\delta/(1+\delta)\) and fire sale losses are 0. Between these two extreme pairs of liquidity generation and fire sale losses, the set of efficient combinations of the two can be calculated. The following proposition addresses the question whether the bank’s strategy should foresee to fire sale the most liquid assets and pledge the rest with the central bank, or the other way round. 
\begin{prop}
\label{prop4}
 In funding strategies to address withdrawals of short-term deposits relying on both funding sources, the bank should always foresee to fire sale the most liquid assets and pledge the less liquid assets. 
\end{prop}
The proof of this proposition is provided in \cite{Bindseil}. The proof relies on showing that with the strategy to fire sale the most liquid assets and pledge the rest, the bank can achieve combinations of liquidity generation and fire sale cost, which are always superior to the combinations under the reverse strategy. The following Proposition \ref{prop5} provides the condition in the case of strategies relying on both funding sources for a SNNR equilibrium, depending on the initial liability structure of the bank and the parameters \(\Theta\) and \(\delta\). 
\begin{prop}
\label{prop5}
Let \(z \in [0,1]\) determine which share of its assets is foreseen by the bank to be used for fire sales (i.e. the less liquid share \(1-z\) of assets are foreseen for pledging with the central bank). Let \(k=h(z)\) be the fire sale losses from fire selling the \(z\) most liquid assets and let \(y=f(z)\) be the total liquidity generated from fire selling the most liquid assets \(z\) and from pledging the least liquid assets \(1-z\). Then a SNNR equilibrium exists if and only if \(\exists z \in [0.1]\) such that
\[ y = f(z) = \frac{\delta}{(\delta +1)}+ \frac{z^{\delta +1}}{(\delta+1)} - \frac{z^{\Theta +1}}{(\Theta+1)}\geq \frac{(1-t-e)}{2}, \qquad
 k=\frac{z^{\Theta +1}}{(\Theta+1)}\leq e.\]
\end{prop}
The proof of this proposition is provided \cite{Bindseil}. Figure \ref{figure 4} illustrates the generation of liquidity and fire sale losses under strategy \(z\). The figure reflects that the bank plans to fire sale the most liquid part of its assets \(z\), and pledge with the central bank the least liquid part of assets \(1-z\). Therefore, total liquidity \(y\) that could be generated corresponds to the sum of \(y_1\), the surface above the fire sale loss curve \(x^{\Theta}\) up to \(z\), and \(y_2\), the surface above the haircut curve \(x^{\delta}\), starting at \(z\). Fire sale losses \(k\) will be equal to the surface below the fire sale loss curve between 0 and \(z\).   

% The following proposition 6 describes the nature of the combinations of \(y\) and \(k\) that can be achieved by varying \(z\) between 0 and 1. This proposition will be the basis for finding an optimal liability structure of the bank as discussed in the subsequent two sections. The optimal strategy \(z\) is determined by the idea to start from the least liquid assets and pledge with the central bank everything as collateral until one needs to switch in order to achieve the necessary total liquidity \(y\).
\begin{figure}[H]
\includegraphics[scale = 0.725]{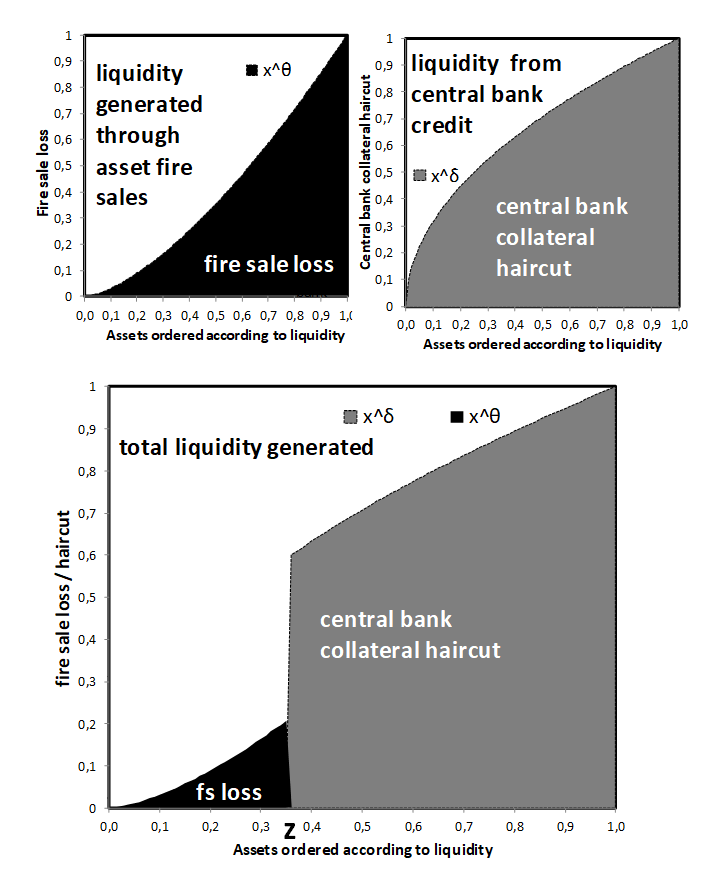}
\caption{Liquidity generation and fire sale losses if the most liquid assets (from \(0\) to \(z\)) are fire sold, while the least liquid assets (from \(z\) to \(1\)) are pledged with the central bank.}
\label{figure 4}
\end{figure}
\section{Stable funding structure with the lowest possible cost} 
In the previous sections, it was assumed that the initial bank balance sheet was given, and the conditions for stability of short-term funding were established. It was shown that depending on the bank’s  liability structure, the haircut \(\delta\) and asset liquidity \(\Theta\), short term bank funding was stable or not. This section will make the liability structure endogenous in a simple setting. It is assumed that different liabilities require different remuneration rates but are at these rates perfectly elastic. 

For given, deterministic \(\delta\) and \(\Theta\), competing banks can be assumed to choose the cheapest possible liability structure as determined by the conditions in the strategic depositor game, such that the single no-run equilibrium applies. Assume that the cost of remuneration of the three asset types are \(r_e\) for equity, \(r_t\) for term funding, and \(0\) for short term deposits. Also assume that \(r_e > r_t > 0\), and that \(\Theta > \delta\). What will in this setting be the composition of the banks’ liabilities?  The objective of the liability structure will be to minimize the average funding cost subject to maintaining a stable short term funding basis. 
\begin{itemize}
	\item One strategy could be to rely only on central bank credit and thus aim to have \(\delta/(\delta+1) \geq (1-e-t)/2\), such that fire sales will not be needed at all as backstop. If fire sales are not needed, then term funding is superior to equity and equity will be set to zero, i.e. liabilities will consist only in term funding \(t\) and short term deposits \(s = 1-t\). Therefore the condition for stable short term funding will be \(\delta/(\delta+1) \geq (1-t)/2\)  which implies \(t^{\ast} =  1-2\delta/(\delta+1)\). The average remuneration rate of bank funding would be \(t^{\ast} r_t\). 
	\item A second strategy would be to rely only on the fire sales approach but to hold the necessary equity. This would mean that the two minimum conditions to be fulfilled are \(z-z^{\Theta}/(1+\Theta) \geq  (1-t-e)/2\) and \(e \geq (z^{(1+\Theta)} /(1+\Theta)\). The funding costs \(tr_t + er_e\)  can be minimized subject to these two constraints to obtain a unique optimum \(t^{\ast}(\Theta, r_e , r_t)\) and \(e^{\ast}(\Theta, r_e , r_t)\), to obtain the minimum average funding costs of the bank liabilities \(t^{\ast} r_t + e^{\ast} r_e\). This case correspond to \(\delta = 0 \) in the following more general formulation.
	\item A third, general strategy is to rely on both sources of funding, and to foresee to fire sell the most liquid assets from \(0\) to \(z\), and to pledge the least liquid assets from \(z\) to \(1\), as described in Propositions \ref{prop4} and \ref{prop5} and in the following optimization problem.
\end{itemize}

\begin{problem}
\label{problem}
Suppose a bank relies on both sources of liquidity.  The general problem of optimal liquidity management is to minimise through the choice of \(t,e,z \in [0,1]\) with \(t+e \in [0,1]\) the average remuneration rate of the banks’ liabilities \(t^{\ast} r_t + e^{\ast} r_e\), with parameters subject to the constraints 
\begin{equation}
\label{suff liquidity condition}
\frac{\delta}{(\delta +1)} + \frac{z^{\delta +1}}{(\delta+1)} -\frac{z^{\Theta +1}}{(\Theta+1)} \geq \frac{(1-t-e)}{2}
\end{equation}
and
\begin{equation}
\label{suff equity condition}
e \geq \frac{z^{\Theta +1}}{(\Theta+1)}.
\end{equation}
\end{problem}
We provide the mathematical solution to this optimization problem in the Appendix \ref{solution}. Appendix \ref{tables and charts} illustrates the functional relationships between the given parameters \((\Theta,\delta,r_e,r_t)\), and the optimum values of \(e^{\ast},t^{\ast},s^{\ast},z^{\ast},r^{\ast}\) (with \(s^{\ast}=1-e^{\ast}-t^{\ast}\) being the implied share of short term funding). Key findings are: the minimum bank intermediation cost \(r^{\ast}\) falls monotonously in \(\Theta\) and in \(\delta\) before reaching zero (Figure A.7). Both the share of assets foreseen for fire selling, \(z^{\ast}\), and the equity ratio fall monotonously in \(\delta\), but both first increase and then decrease again in \(\Theta\) (unless \(\delta\) is high enough to allow for \(z^{\ast}=0\) and \(e^{\ast}=0\)) (Figure A.8 and A.9). The share of long term debt, \(t^{\ast}\), falls monotonously in \(\Theta\) and in \(\delta\) before reaching zero (Figure A.10). The share of short term deposits, \(s^{\ast}\), increases monotonously in \(\Theta\) and in \(\delta\) before reaching 100\% (Figure A.11).
Figure A.3 in the Appendix shows how the liability mix minimizing the funding cost evolves as a function of asset liquidity \(\Theta\), for given collateral liquidity \(\delta = 0.2\) (and again \(r_t = 5\%, \ r_e  = 10\%\)).  It illustrates that the optimal equity share is not a monotonous function of the asset liquidity, but reaches a maximum for around \(\Theta = 0.9\). In contrast, the optimal long and short-term funding shares decline monotonously with an improving asset liquidity. (Figure A.1 shows the same relations, but for \(\delta = 0\))  
Figure A.4 illustrates how the liability mix minimizing funding cost evolves as function of the equity premium \(r_e\), for a given term funding premium \(r_t = 2\%\), for given \((\Theta, \delta) = (0.5, 0.2)\). The equity share \(e^{\ast}\) falls monotonously in the equity premium \(r_e\), while the share of long term debt \(t^{\ast}\) increases monotonously in \(r_e\).  The share of short term funding \(s^{\ast}\) also declines monotonously in \(r_e\). (Figure A.2 shows the same relationships for \(\delta = 0\)).

\section{The central bank collateral framework as a policy tool}
While the primary purpose of the central bank collateral framework is risk protection, the observed collateral policy measures of central banks in 2008/2009 and in 2020 raise the question what exactly the intentions of the central banks have been to widen collateral availability (and hence potential central bank recourse) in particular in a context of deteriorating asset liquidity. The model proposed in this paper allows interpreting the relaxation of the collateral framework as a policy measure:
\begin{enumerate}
	\item First, when \(\Theta\) (asset liquidity) suddenly declines, increasing \(\delta\) (by decreasing haircuts and broadening collateral eligibility) is a way to preserve the no-run equilibrium (the SNNR) and thereby is a necessary condition to prevent increases in central bank reliance, fire sales, and/or defaults. In this sense, it benefits all banks and financial stability in general, and not only those banks who already experience an actual run. Moreover, it may be noted that the model provides support to Bagehot’s (\cite{Bagehot}) ``inertia principle'' according to which the central bank should not tighten its collateral framework in a financial crisis as a reaction to the deterioration of asset liquidity: ``If it is known that the Bank of England is freely advancing on what in ordinary times is reckoned a good security—on what is then commonly pledged and easily convertible—the alarm of the solvent merchants and bankers will be stayed…''. Lowering \(\delta\) when anyway \(\Theta\) declines would mean to decrease particularly strongly the amount of sustainable short term funding and thereby to maximize the probability of a destabilization of bank funding, contributing, instead of preventing, large central bank recourse and fire sales of assets.
\item Second, assuming that a deterioration of \(\Theta\) can be anticipated as a crisis is building up, one could imagine that banks can adjust their liability structure in time. It would come at a high cost because in such a context also investors will have a strong preference for short term assets and the collective attempt of all banks to increase the maturity of their liabilities will therefore lead to a steep increase of bank funding costs (and hence of bank lending rates). This would be pro-cyclical, and an adjustment of the collateral framework parameter \(\delta\) could be seen as a policy tool to prevent such a steep increase of funding costs. 
\item Third, while the effect described in the previous point could in theory also be addressed by conventional monetary policy, i.e. a lowering of central bank interest rates, this has limits as far as the zero lower bound is reached (as it is the case for most central banks in 2020). When this limit is reached, then a widening of collateral availability may become relevant as an alternative approach to lowering effective bank funding costs or at least prevent their increase. 

In the context of the model, assume that \((\Theta_1, \delta_1)\) are the pre-crisis asset parameters and (for a given equity and term funding premia \(r_t, r_e)\)  the minimum funding cost is \(r_1^{\ast} = r^{\ast}(\Theta_1, \delta_1)\), and the related liability composition is \(e_1^{\ast} = e^{\ast}(\Theta_1, \delta_1), \ t_1 = t^{\ast}(\Theta_1, \delta_1), \ s_1^{\ast} = s^{\ast}(\Theta_1, \delta_1)\). Assume that a market liquidity crisis materializes with \(\Theta\) shifting from \(\Theta_1\) to \(\Theta_2 <\Theta_1\). This implies that \(r_2^{\ast} = r^{\ast}(\Theta_2,\delta_1) > r_1^{\ast}\). Moreover, for \((\Theta_2, \delta_1)\), the funding structure \(e_1^{\ast}, t_1^{\ast},s_1^{\ast}\) does not allow for a single no-run equilibrium, but the banks are now in the multiple equilibrium case in which a bank run can occur. Now call \(\delta_2\) the value of \(\delta\) for which \(r^{\ast}(\Theta_2, \delta_2)= r_1^{\ast}\). In words, and assuming that the central bank is constrained by the zero-lower bound, the collateral framework \(\delta_2\) is the one that allows to restore the monetary conditions (in the sense of the bank funding costs, and thus the bank lending rates) that prevailed before the liquidity crisis, however only under the assumption that the banks can rapidly adjust their funding structure towards \(e_2^{\ast} = e^{\ast}(\Theta_2,\delta_2), \ t_2^{\ast} = t^{\ast}(\Theta_2,\delta_2), \ s_2^{\ast} = s^{\ast}(\Theta_2,\delta_2)\). Then, monetary conditions, and financial stability have been restored. Alternatively, the central bank may immediately want to restore the single no-run equilibrium and not take the risk of a run in the possibly lengthy process of the adjustment of banks’ liability structure. This will require choosing another value of \(\delta \geq \delta_2\).  Call \(\delta_3\) the collateral framework which is obviously sufficient to restore immediately a no-run equilibrium, in the sense that \( e_1^{\ast} \leq e^{\ast}(\Theta_2,\delta_3), \ t_1^{\ast} \leq t^{\ast}(\Theta_2,\delta_3), \ s_1^{\ast} \geq s^{\ast}(\Theta_2,\delta_3)\). However, \(\delta_3\) is not strictly necessary for restoring a single no-run equilibrium. Call \(\delta_4\) the smallest value of \(\delta\) for which the funding structure \((e_1^{\ast}, t_1^{\ast}, s_1^{\ast})\) implies a no-run equilibrium for \((\Theta_2,\delta_4)\), i.e. for which the sufficient equity condition \eqref{suff equity condition} and the sufficient liquidity condition \eqref{suff liquidity condition} are both fulfilled, whereby banks can of course recalibrate the parameter \(z\) without any delay.

In what follows, we refer the reader to the example in Appendix \ref{tables and charts}, with \(r_t = 5\%\) and \(r_e =10\%\).  Short-term funding costs are set at 0\%, consistent with the idea that the zero lower bound is constraining. Table A.1 and figures to A.11 show, for these given funding cost parameters, how the optimum values \(r^{\ast}\), \(z^{\ast}\), \(e^{\ast}\), \(t^{\ast}\) and \(s^{\ast}\) depend on the asset fire sale liquidity parameter \(\Theta\) and the collateral framework parameter \(\delta\). 
 
Assume that \((\Theta_1, \delta_1) = (0.7, 0.2)\), and that the banks have chosen the cheapest sustainable funding structure. Following Table A.1, we have \((e_1^{\ast}, t_1^{\ast}, s_1^{\ast})= (67\%, 15\%, 19\%)\), so as to obtain funding costs of \(r_1^{\ast}=2.5\%\). Let the crisis-related asset liquidity be \(\Theta_2=0.2\). The table indicates that \(\delta_2 = 0.4\), as \(r^{\ast}( \Theta_2, \delta_2) = 2.2\% \leq r_1^{\ast}\). A more radical move by the central bank to \(\delta_3=0.7\) also obviously and immediately solves the issue of the bank run, as the optimal and thus sustainable funding mix for \((\Theta_2,\delta_3) = (0.2, 0.7)\) is less demanding for sure than the pre-crisis funding mix, since the table indicates that \(e_1^{\ast} \leq  e^{\ast}(\Theta_2, \delta_3), \ t_1^{\ast} \leq  t^{\ast}(\Theta_2, \delta_3),\ s_1^{\ast}\geq s^{\ast}(\Theta_2, \delta_3)\). Finally, one may check if \(\delta = 0.6\) or even \(\delta = 0.5\) would also ensure the fulfilment of the sufficient equity and sufficient liquidity conditions \eqref{suff equity condition} and \eqref{suff liquidity condition}. Note that for \((\Theta_1, \delta_1), z^{\ast}=44\%\), while for \((\Theta_2,\delta_3)\), like for any \((\Theta, \delta)\) with \(\Theta \leq \delta, \ z^{\ast} = 0\). It turns out that the bank simply needs to set \(z=0\) and both \(\delta_4 =0.6\) and \(\delta_4 = 0.5\) become sufficient from a funding stability perspective, while \(\delta_4 =0.4\) is indeed insufficient. Therefore, to immediately restore both an accommodative monetary policy stance and financial stability, the central bank in this case should move after the liquidity crisis (after the deterioration of \(\Theta\) from \(0.7\) to \(0.2\)) its collateral framework \(\delta\) from \(0.2\) to \(0.5\). If this would be too accommodating from the monetary policy perspective, then the central bank could of course raise the short-term risk free interest rate.

Figure A.6 in the Appendix shows how the liability mix minimizing the funding cost evolves as a function of asset liquidity \(\Theta\), for given collateral liquidity \(\delta=0.2\) (and again \(r_t = 5\%, \ r_e = 10\%\)). It illustrates that the optimal equity share is not a monotonous function of the asset liquidity, but reaches a maximum for around \(\Theta=0.9\). In contrast, the optimal long- and short-term funding shares decline monotonously with an improving asset liquidity. 

Figure A.7 illustrates how the liability mix minimizing funding cost evolves as function of the equity premium \(r_e\), for a given term funding premium \(r_t=2\%\), for given \((\Theta, \delta) = (0.5, 0.2)\). The equity share \(e^{\ast}\) falls monotonously in the equity premium \(r_e\), while the share of long term debt  \(t^{\ast}\) increases monotonously in \(r_e\). The share of short term funding \(s^{\ast}\) also declines monotonously in \(r_e\). 
\end{enumerate}
\section{Conclusion}
This paper provides a model of the role of the central bank collateral framework for the LOLR, bank intermediation costs, banks’ optimal funding structure, and monetary policy at the zero lower bound. Its innovation relative to the existing literature consists in specifying asset liquidity and the collateral framework as continuous functions in the unit interval, which allows building an integrated model encompassing all of these dimensions, and deriving an analytical solution of this model. It was shown how broadening the collateral framework, such as done in e.g. 2008 and 2020 by a number of central banks, can be considered to restore financial stability and adequate monetary conditions after a negative asset liquidity shock. 
Bank asset liquidity and the central bank collateral framework jointly determine financial stability and the ability of the banking system to deliver maturity transformation, which is one of its key functions. 
The model also shows that a widening of the central bank collateral set can \emph{prevent} large recourse to central bank credit by banks suffering from a deterioration of asset liquidity. In this sense, the paper provides further illustration of Bagehot’s (\cite{Bagehot}) conjecture that only the ``brave plan'' of the 19th century Bank of England would be a ``safe'' plan. In other words, by being ``brave'' and increasing \(\delta\) after an exogenous drop of \(\Theta\), the central bank will preserve the banks’ funding stability and thereby minimize the recourse of stressed banks to its credit facilities, and hence be on the ``safe'' side also in terms of financial exposures). 
The model also allowed to identify the impact of asset liquidity and of the central bank collateral framework on funding costs of banks and thereby on monetary policy conditions: first, policy makers need to be aware that a tightening of any of the two emergency liquidity sources of banks needs to be, everything else unchanged, compensated by a lowering of the monetary policy interest rate to maintain unchanged funding costs of the real economy. Second, when the central bank has reached the zero lower bound, and therefore cannot use standard interest rate policies any longer, it can consider to use its collateral framework to counteract a further increase of actual funding costs of banks (and hence of the real economy depending on banks) which would otherwise result from the deteriorated asset liquidity, because of the implied more expensive funding structure.

\appendix
\label{tables and charts}
% \begin{figure}[h]
% \includegraphics[scale = 0.25]{Annex_table_1.png}
% \end{figure}
\includepdf[scale = 0.85, pages=1,pagecommand=\section{Tables and Charts}]{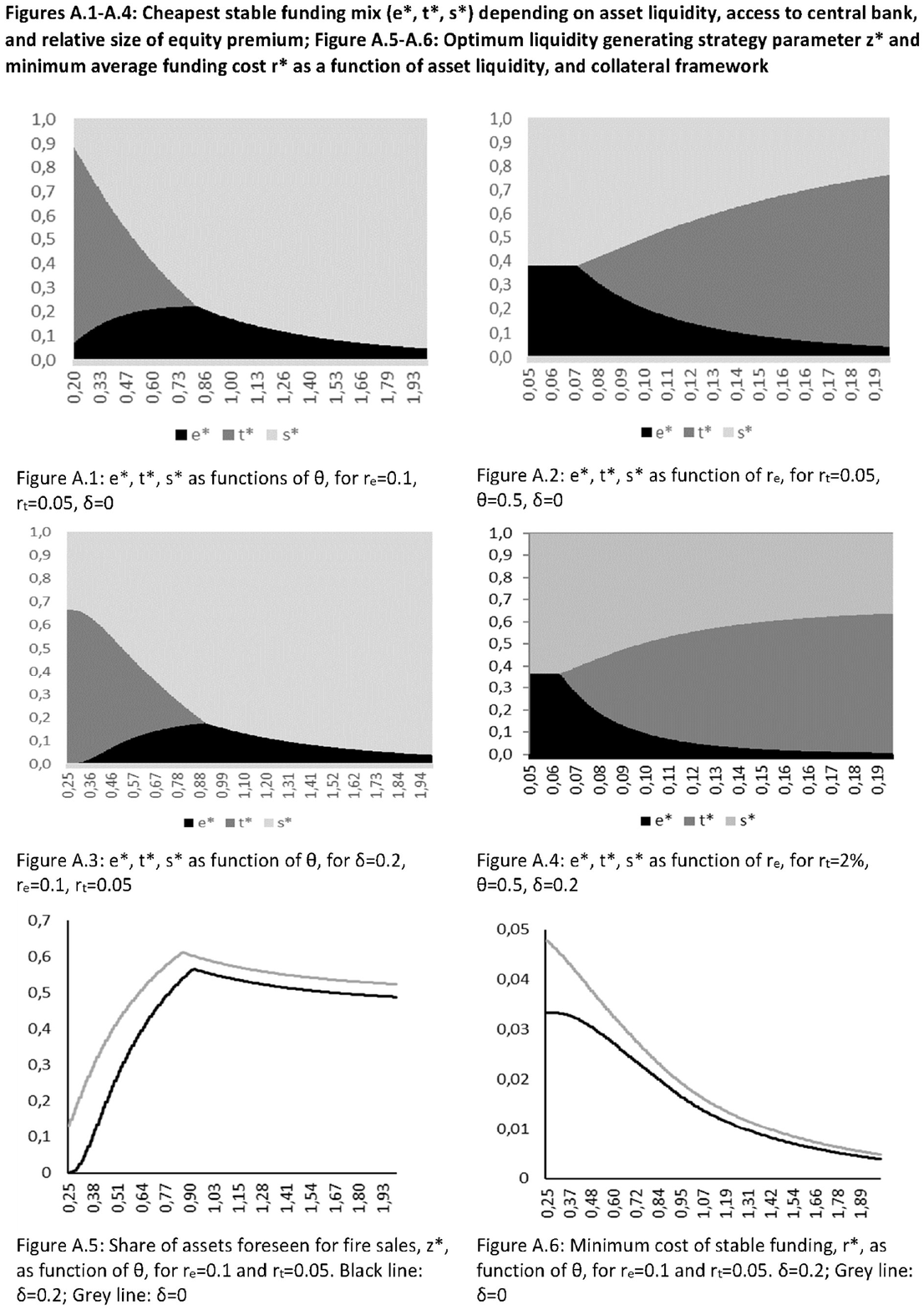}
\includepdf[scale = 0.85, pages=2-,pagecommand={}]{charts.pdf}
% \includepdf[pages=-]{charts.pdf}
\section{Solution to Problem \ref{problem}}
\label{solution}
The problem we want to solve consists of minimizing the function \(r\colon (t,e) \mapsto r_t t + r_e e\) for given parameters \(r_e \geq r_t\), over the subspace of triples \((t,e,z) \in [0,1]^3\), subject to the following constraints:
\begin{equation}
\tag{1}
	\frac{\delta}{\delta +1} + \frac{z^{\delta +1}}{\delta+1} -\frac{z^{\Theta +1}}{\Theta+1} \geq \frac{(1-t-e)}{2}
\end{equation}
\begin{equation}
\tag{2}
e \geq \frac{z^{\Theta +1}}{\Theta+1}	
\end{equation}
\begin{equation}
\tag{3}
t +e \leq 1.
\end{equation}
where \(\Theta > \delta\geq 0\) are given parameters.

Let \(W \subset [0,1]^3\) be the domain of admissible triples \((t,e,z)\). Clearly, \(W\) is a compact subspace of \(\R^3\), therefore \(r \colon W \to \R\) admits a minimum. Since the gradient \(\nabla r = (r_t,r_e,0)\) is non-vanishing everywhere on \(W\) the minimum point lies on the boundary \(\partial W\). We now analyze the behaviour of \(r\) over a decomposition of such boundary. Note that \(\partial W = \cup_{i=1}^{5} W_i\), where we set:
\begin{equation}
	\tag{1}
	W_1 = \left\{(t,e,z) \in W \colon e= 0 \text{, or } t= 0 \text{, or } z = 0\right\}
\end{equation}

\begin{equation}
	\tag{2}
	W_2 = \left\{(t,e,z) \in W \colon e= 1 \text{, or } t= 1 \text{, or } z = 1\right\}
\end{equation}

\begin{equation}
	\tag{3}
	W_3 = \{(t,e,z) \in W \colon t+e = 1\}
\end{equation}

\begin{equation}
	\tag{4}
	W_4 = \left\{(t,e,z) \in W \colon e = \frac{z^{\Theta +1}}{\Theta +1}\right\}
\end{equation}

\begin{equation}
	\tag{5}
	W_5 = \left\{(t,e,z) \in W \colon \frac{\delta}{\delta +1} + \frac{z^{\delta +1}}{\delta+1} -\frac{z^{\Theta +1}}{\Theta+1} = \frac{(1-t-e)}{2}\right\}
\end{equation}
Let's study \(r\) on \(W_1\). If \(e = 0 \) then necessarsily \(z= 0 \), which implies:
\[\frac{\delta}{\delta +1} \geq \frac{(1-t)}{2}\]
If \(\delta \leq 1\) we minimize this by choosing \(\left(t^{\ast},e^{\ast},z^{\ast}\right) = \left(\frac{1-\delta}{1+\delta}, 0 ,0\right)\), which yields \(r^{\ast}=r_t \frac{1-\delta}{1+\delta}\). Note in this case \(t \leq 1\), and equality holds if and only if \(\delta = 1\). If instead \(\delta > 1\), then \((t^{\ast},e^{\ast},z^{\ast}) = (0,0,0)\) is admissible, so we obtain \(r^{\ast} = 0\). From now on we will assume \(\delta <	 1\) to avoid the trivial solution.

Suppose \(t = 0\), then we set \(e = \frac{z^{\Theta +1}}{\Theta +1}\), subject to the condition \[h(z) = \frac{\delta -1}{\delta + 1} + 2 \frac{z^{\delta +1}}{\delta +1} - \frac{z^{\Theta +1}}{\Theta +1} \geq 0.\]
We have \(h(0) <0\), \(h(1) >0\) and \(h'(z) > 0\) for every \(z \in [0,1]\). Therefore \(\exists! z^{\ast} \in [0,1]\) with \(h(z^{\ast}) = 0\). By construction, \(z^{\ast}\) is the smallest value of \(z\) for which \(\left(0,\frac{z^{\Theta +1}}{\Theta +1},z\right)\) satisfies all constraints. Thus we get a candidate triple \(\left(0,\frac{(z^{\ast})^{\Theta +1}}{\Theta +1},z^{\ast}\right)\), which yields \(r^{\ast} = r_e \frac{(z^{\ast})^{\Theta +1}}{(\Theta +1)} \).
Finally, \(z = 0 \) can be reduced to an already analyzed case.

On \(W_2\) and \(W_3\) it is easy to prove there is no better contribute from any possible configuration. Turning to \(W_4\), we have \(e = \frac{z^{\Theta +1}}{\Theta +1}\) and \(t\) is subject to the constraint:
\[t \geq \frac{1-\delta}{1+\delta} + \frac{z^{\Theta +1}}{\Theta + 1} -2 \frac{z^{\delta +1}}{\delta +1} = -h(z)\]
To minimize we set \(t = -h(z)\), and we study \(r = \phi(z) = -r_th(z) + r_e \left(\frac{z^{\Theta +1}}{\Theta +1}\right)\). By solving \(\phi '(z) = 0\), we see that \(\phi\) has a minimum at:
\[\overline{z}  = \left(\frac{2r_t}{r_t + r_e}\right)^{\frac{1}{\Theta - \delta}}.\]
If \(h(\overline{z}) \leq 0\) we have a new candidate triple \((t^{\ast},e^{\ast},z^{\ast}) = \left(-h(\overline{z}),\frac{\overline{z}^{\Theta +1}}{\Theta +1},\overline{z}\right)\). Finally, the case of \(W_5\) can be reconducted to this last one, so the analysis is complete. 

The solution to the original problem is obtained by comparing the two (or three, depending on the cases) candidates selected by the algorithm we have just described.
\begin{rmk}
If we fix \(\delta\) and let \(\Theta\) vary, we can observe that:
\begin{enumerate}
	\item \(z^{\ast} ( \Theta) \leq z^{\ast} ( \Theta')\) if \(\Theta > \Theta '\), where \(z^{\ast}(\Theta)\) denotes the unique solution to \( \frac{\delta -1}{\delta + 1} + 2 \frac{z^{\delta +1}}{\delta +1} - \frac{z^{\Theta +1}}{\Theta +1} = 0\) over \([0,1]\).
	\item \(\lim_{\Theta \to +\infty} \overline{z}(\Theta) = 1\), for \( \overline{z}(\Theta) = \left(\frac{2r_t}{r_t + r_e}\right)^{\frac{1}{\Theta - \delta}}\).
\end{enumerate}
Thus, for \(\Theta -\delta\) big enough, we obtain \(\overline{z}(\Theta) > z^{\ast}(\Theta)\), which forces us to pick \((t^{\ast},e^{\ast},z^{\ast}) = \left(0,\frac{z^{\ast}(\Theta)}{\Theta +1},z^{\ast}(\Theta)\right)\) over \((t^{\ast},e^{\ast},z^{\ast}) = \left(-h(\overline{z}),\frac{\overline{z}^{\Theta +1}}{\Theta +1},\overline{z}\right)\).
On the other hand, for \(\Theta\) close enough to \(\delta\), we have that \((t^{\ast},e^{\ast},z^{\ast}) = \left(-h(\overline{z}),\frac{\overline{z}^{\Theta +1}}{\Theta +1},\overline{z}\right)\) is admissible, so that \((t^{\ast},e^{\ast},z^{\ast}) = \left(0,\frac{z^{\ast}(\Theta)}{\Theta +1},z^{\ast}(\Theta)\right)\) no longer runs among the candidates for the minimization of \(r\).
\end{rmk}
% Bibliography
% %
\bibliographystyle{amsplain}
\bibliography{bibliography}

\providecommand{\bysame}{\leavevmode\hbox to3em{\hrulefill}\thinspace}
\providecommand{\MR}{\relax\ifhmode\unskip\space\fi MR }
\providecommand{\MRhref}[2]{%
  \href{http://www.ams.org/mathscinet-getitem?mr=#1}{#2}
}
\providecommand{\href}[2]{#2}
\begin{thebibliography}{10}

\bibitem{Acharya-Gale}
V.V {Acharya}, D.~{Gale}, and T.~{Yorulmazer}, \emph{Rollover risk and market
  freezes}, Journal of Finance \textbf{66} (2011), 1177--1209.

\bibitem{Acharya}
V.V. {Acharya} and S.~{Wiswanathan}, \emph{Leverage, moral hazard, and
  liquidity}, Journal of Finance \textbf{66} (2011), 99--138.

\bibitem{Aco}
J.~{Acosta-Smith}, G.~Arnould, K.~Milonas, and Q.-A. Vo, \emph{Bank capital and
  liquidity transformation}, Working paper (2018).

\bibitem{Ashcraft-Garleanu}
A.~{Ashcraft}, N.~{Garleanu}, and L.H. {Pedersen}, \emph{Two monetary tools:
  Interest rates and haircuts}, NBER Macroeconomics Annual \textbf{25} (2010),
  143--180.

\bibitem{Bagehot}
W.~{Bagehot}, \emph{Lombard street: A description of the money market},
  (1873).

\bibitem{Basel}
{Basel Committee on Banking Supervision}, \emph{Basel iii - the liquidity
  coverage ratio and liquidity risk monitoring tools}, Basel Committee on
  Banking Supervision (2013).

\bibitem{Bindseil}
U.~{Bindseil}, \emph{Central bank collateral, asset fire sales, regulation and
  liquidity}, European Central Bank, Working Paper Series (2013), no.~1610.

\bibitem{Bindseil2017}
U.~{Bindseil}, M.~{Corsi}, B.~{Sahel}, and A.~{Visser}, \emph{The ecb
  collateral framework explained}, ECB Occasional Paper (2017).

\bibitem{Bindseil2017-2}
U.~{Bindseil}, G.~{Dragu}, A.~{D\"uring}, and J.~{von Landesberger},
  \emph{Asset liquidity, central bank collateral, and banks’ liability
  structure}, preprint (2017),
  \url{https://papers.ssrn.com/sol3/papers.cfm?abstract_id=3049792}.

\bibitem{boy}
N.~{Boyson}, J.~Helwege, and J.~Jindra, \emph{Crises, liquidity shocks, and
  fire sales at commercial banks}, Financial Management \textbf{43} (2014),
  857--884.

\bibitem{Brunnermeier}
M.~{Brunnermeier} and L.H. {Pedersen}, \emph{Market liquidity and funding
  liquidity}, Review of Financial Studies \textbf{22} (2009), 2201--2238.

\bibitem{carl}
M.~{Carlson}, B.~Duygan-Bump, and W.~Nelson, \emph{Why do we need both
  liquidity regulations and a lender of last resort? a perspective from federal
  reserve lending during the 2007-09 u.s. financial crisis}, Finance and
  Economics Discussion Series 2015-011. Washington: Board of Governors of the
  Federal Reserve System (2015).

\bibitem{Chapman}
J.T.E {Chapman}, J.~{Chiu}, and M.~{Molico}, \emph{Central bank haircut
  policy}, Annals of Finance \textbf{7} (2011), 319--348.

\bibitem{Davy}
S.A. {Davydenko}, I.A. {Strebulaev}, and X.~{Zhao}, \emph{A market-based study
  of the cost of default}, Review of Financial Studies \textbf{25} (2011),
  2959–2999.

\bibitem{Diamond}
D.W. {Diamond} and P.H. {Dybvig}, \emph{Bank runs, deposit insurance, and
  liquidity}, Journal of Political Economy \textbf{91} (1983), 401--419.

\bibitem{Dotz}
N.~{D\"otz} and M.~{Weth}, \emph{Redemptions and asset liquidations in
  corporate bond funds}, Bundesbank Discussion paper (2019).

\bibitem{ECB20}
ECB, \emph{Ecb announces package of temporary collateral easing measures}, ECB
  Press Release (2020).

\bibitem{FSA}
{Financial Services Authority}, \emph{The turner review - a regulatory response
  to the global banking crisis},  (2009).

\bibitem{DeFiore}
{De Fiore} Fiorella, Hoerova Marie, and Uhlig Harald, \emph{Money markets,
  collateral and monetary policy}, Becker-Friedman Institute{,} University of
  Chicago (2018).

\bibitem{Fudenberg}
D.~{Fudenberg} and J.~{Tirole}, \emph{Game theory}, MIT Press, 1991.

\bibitem{Gertler}
Mark {Gertler} and Kiyotaki {Nobuhiro}, \emph{Banking, liquidity, and bank runs
  in an infinite horizon economy}, American Economic Review \textbf{105}
  (2015), 2011--43.

\bibitem{Glover}
B.~{Glover}, \emph{The expected costs of default}, working paper, Carnegie
  Mellon University (2011).

\bibitem{Markets}
{Markets Committee}, \emph{Central bank collateral frameworks and practices},
  Report by a Study Group chaired by Guy Debelle and established by the BIS
  Markets Committee. (2013).

\bibitem{Morris}
S.~{Morris} and H.S. {Shin}, \emph{The theory of global games}, Cowles
  Foundation Discussion Papers \textbf{1275R} (2001).

\bibitem{Nyborg}
K.~{Nyborg}, \emph{Collateral frameworks: The open secret of central banks},
  Cambridge University Press, 2016.

\bibitem{Rochet}
J.-C. {Rochet} and X.~{Vives}, \emph{Coordination failures and the lender of
  last resort: was bagehot right after all?}, Journal of the European Economic
  Association \textbf{2} (2004), 1116--1147.

\bibitem{Shlei}
A.~{Shleifer} and R.~Vishny, \emph{Fire sales in finance and macroeconomics},
  Journal of Economic Perspectives \textbf{25} (2011), no.~1, 29--48.

\end{thebibliography}
\end{document}